\documentclass[manuscript,11pt]{acmart}

\copyrightyear{}
\acmYear{}
\setcopyright{none}
\acmConference[]{}{}{}
\acmBooktitle{}
\acmPrice{}
\acmDOI{}
\acmISBN{}

\usepackage{fullpage}
\usepackage{float}
\usepackage{slivkins-setup,slivkins-theorems}
\usepackage{amsbsy, amscd, bm, bbm}
\usepackage{array}
\usepackage{booktabs}
\usepackage{graphicx}
\usepackage{color}
\usepackage{ifthen}
\usepackage{xspace}
\usepackage[noend]{algorithmic}
\usepackage{algorithm}
\newcommand{\TAB}{\hspace{5mm}} 

\usepackage{enumerate}
\usepackage{mdframed}
\usepackage{comment}

\usepackage{color-edits}
\addauthor{as}{red}      
\addauthor{jm}{blue}     
\addauthor{ni}{green}     
\addauthor{sw}{magenta}     

\newcommand{\term}[1]{\ensuremath{\mathtt{#1}}\xspace}
\newcommand{\thread}{\term{thread}}
\newcommand{\ALG}{\term{ALG}}

\newcommand{\DT}{\mD} 
\newcommand{\prior}{\mP} 

\def\DKL{\textbf{D}_{\term{KL}}}

\def\E{\mathop{\mathbb{E}}}
\def\I{\mathcal{I}}

\def\A{\mathcal{A}}
\def\M{\mathcal{M}}
\def\S{\mathcal{S}}
\def\X{\mathcal{X}}
\def\EX{\term{EX}}

\def\EE{\mathcal{E}}
\def\varTheta{\bold{\Theta}}
\def\varOmega{\bold{\Omega}}
\def\ell{l}

\newcommand{\AExp}{\mA} 
\newcommand{\MExp}{\mM}  

\def\OPT{\term{OPT}}
\newcommand{\OPTpub}{\OPT_{\term{pub}}}
\newcommand{\OPTpri}{\OPT_{\term{pri}}}

\newcommand{\pRes}{p^{\term{res}}}

\begin{document}

\title{Bayesian Exploration with Heterogeneous Agents}

\author{Nicole Immorlica}
\affiliation{\institution{Microsoft Research}}
\email{nicimm@microsoft.com}

\author{Jieming Mao}
\affiliation{\institution{University of Pennsylvania}}
\email{maojm517@gmail.com}

\author{Aleksandrs Slivkins}
\affiliation{\institution{Microsoft Research}}
\email{slivkins@microsoft.com}

\author{Zhiwei Steven Wu}
\affiliation{\institution{University of Minnesota}}
\email{zsw@umn.edu}

\begin{abstract}
It is common in recommendation systems that users both consume and produce information as they make strategic choices under uncertainty. While a social planner would balance ``exploration'' and ``exploitation'' using a multi-armed bandit algorithm, users' incentives may tilt this balance in favor of exploitation. We consider Bayesian Exploration: a simple model in which the recommendation system (the ``principal'') controls the information flow to the users (the ``agents'') and strives to incentivize exploration via information asymmetry. A single round of this model is a version of a well-known ``Bayesian Persuasion game'' from \cite{Kamenica-aer11}. We allow heterogeneous users, relaxing a major assumption from prior work that users have the same preferences from one time step to another. The goal is now to learn the best \emph{personalized} recommendations. One particular challenge is that it may be impossible to incentivize some of the user types to take some of the actions, no matter what the principal does or how much time she has. We consider several versions of the model, depending on whether and when the user types are reported to the principal, and design a near-optimal ``recommendation policy'' for each version. We also investigate how the model choice and the diversity of user types impact the set of actions that can possibly be ``explored'' by each type.

\end{abstract}
%

\maketitle

\section{Introduction}
Recommendation systems are ubiquitous in online markets (\eg Netflix for movies, Amazon for products, Yelp for restaurants, etc.),
high-quality recommendations being a crucial part of their value proposition.
A typical recommendation system encourages its users to submit feedback on their experiences, and aggregates this feedback in order to provide better recommendations in the future. Each user plays a dual rule: she consumes information from the previous users (indirectly, via recommendations),
and produces new information (\eg  a review) that benefits future users. This dual role creates a tension between exploration, exploitation, and users' incentives.

A social planner -- a hypothetical entity that controls users for the sake of common good -- would balance ``exploration'' of insufficiently known alternatives and ``exploitation'' of the information acquired so far. Designing algorithms to trade off these two objectives is a well-researched subject in machine learning and operations research.
%
However, a given user who decides to ``explore'' typically suffers all the downside of this decision, whereas the upside (improved recommendations) is spread over many users in the future. Therefore, users' incentives are skewed in favor of exploitation. As a result, observations may be collected at a slower rate, and suffer from selection bias (\eg  ratings of a particular movie may mostly come from people who like this type of movies). Moreover, in some natural but idealized examples (\eg  \cite{Kremer-JPE14,ICexploration-ec15}) optimal recommendations are never found because they are never explored.

Thus, we have a problem of \emph{incentivizing exploration}.
Providing monetary incentives can be financially or technologically unfeasible, and relying on voluntary exploration can lead to selection biases. A recent line of work, started by \cite{Kremer-JPE14}, relies on the inherent \emph{information asymmetry} between the recommendation system and a user. These papers posit a simple model, termed \emph{Bayesian Exploration} in \cite{ICexplorationGames-ec16}. The recommendation system is a ``principal'' that interacts with a stream of self-interested ``agents'' arriving one by one. Each agent needs to make a decision: take an action from a given set of alternatives. The principal issues a recommendation,
and observes the outcome, but cannot direct the agent to take a particular action. The problem is to design a ``recommendation policy'' for the principal that learns over time to make good recommendations \emph{and} ensures that the agents are incentivized to follow this recommendation.
A single round of this model is a version of a well-known ``Bayesian Persuasion game'' \cite{Kamenica-aer11}.

\xhdr{Our scope.}
We study Bayesian Exploration with agents that can have heterogenous preferences. The preferences of an agent are encapsulated in her {\em type}, \eg vegan vs meat-lover.
When an agent takes a particular action, the outcome depends on the action itself (\eg the selection of restaurant), the ``state'' of the world (\eg the qualities of the restaurants), and the type of the agent. The state is persistent (does not change over time), but initially not known; a Bayesian prior on the state is common knowledge. In each round, the agent type is drawn independently from a fixed and known distribution. The principal strives to learn the best possible recommendation for each agent type.

We consider three models, depending on whether and when the agent type is revealed to the principal: the type is revealed immediately after the agent arrives (\emph{public types}), the type is revealed only after the principal issues a recommendation (\emph{reported types}),\footnote{Reported types may arise if
the principal asks agents to report the type after the recommendation is issued, \eg in a survey. While the agents are allowed to misreport their respective types, they have no incentives to do that.}
and the type is never revealed (\emph{private types}).
We design a near-optimal recommendation policy for each modeling choice.
\OMIT{In fact, we consider a stronger benchmark: optimal Bayesian-expected reward achieved by any recommendation policy in any one round.}

\xhdr{Explorability.}
A distinctive feature of Bayesian Exploration is that it may be impossible to incentivize some agent types to take some actions, no matter what the principal does or how much time she has. For a more precise terminology, a given type-action pair is \emph{explorable} if this agent type takes this action under some recommendation policy in some round with positive probability. This action is also called \emph{explorable} for this type. Thus: some type-action pairs might not be explorable. Moreover, one may need to explore to find out which pairs are explorable.  The set of explorable pairs is interesting in its own right as they bound the welfare of a setting. Recommendation policies cannot do better than the ``best explorable action'' for a particular agent type: an explorable action with a largest reward in the realized state.

\xhdr{Comparative statics for explorability.}
We study how the set of all explorable type-action pairs (\emph{explorable set}) is affected by the model choice and the diversity of types. First, we find that for each problem instance the explorable set stays the same if we transition from public types to reported types, and can only become smaller if we transition from reported types to private types.
We provide a concrete example when the latter transition makes a huge difference. Second, we vary the distribution $\DT$ of agent types. For public types (and therefore also for reported types), we find that the explorable set is determined by the support set of $\DT$. Further, if we make the support set larger, then the explorable set can only become larger. In other words, \emph{diversity of agent types helps exploration}. We provide a concrete example when the explorable set increases very substantially even if the support set increases by a single type. However, for private types the picture is quite different: we provide an example when \emph{diversity hurts}, in the same sense as above. Intuitively, with private types, diversity muddles the information available to the principal making it harder to learn about the state of the world, whereas for public types diversity helps the principal refine her belief about the state.

\xhdr{Our techniques.}
As a warm-up, we first develop a recommendation policy for public types. In the long run, our policy
matches the benchmark of ``best explorable action''.
While it is easy to prove that such a policy exists, the challenge is to provide it as an explicit procedure.
Our policy focuses on exploring all explorable type-action pairs. Exploration needs to proceed gradually, whereby exploring one action may enable the policy to explore another. In fact, exploring some action for one type may enable the policy to explore some action for another type. Our policy proceeds in phases: in each phase, we explore all actions for each type that can be explored using information available at the start of the phase.  Agents of different types learn separately, in per-type ``threads"; the threads exchange information after each phase.

An important building block is the analysis of the single-round game. We use information theory to characterize how much state-relevant information the principal has. In particular, we prove a version of \emph{information-monotonicity}: the set of all explorable type-action pairs can only increase if the principal has more information.

As our main contribution, we develop a policy for private types. In this model, recommending one particular action to the current agent is not very meaningful because the agents' type is not known to the principal. Instead, one can recommend a \emph{menu}: a mapping from agent types to actions. Analogous to the case of public types, we focus on explorable menus and gradually explore all such menus, eventually matching the Bayesian-expected reward of the best explorable menu.
\OMIT{Without loss of generality, we restrict to  Bayesian-incentive compatible (BIC) policies: essentially, policies that output menus such that the agents are incentivized to follow them. The issue of explorability is now about menus: a menu is called \emph{explorable} if some BIC policy recommends this menu in some round with a positive probability. As some menus might not be explorable, we are interested in the ``best explorable menu'': an explorable menu with a largest expected reward for the realized state of the world. Our recommendation policy for private types competes with this benchmark, eventually matching its Bayesian-expected reward. Our policy focuses on exploring all explorable menus, and proceeds gradually, whereby exploring one menu may enable the policy to explore another.}
One difficulty is that exploring a given menu does not immediately reveal the reward of a particular type-action pair (because multiple types could map to the same action). Consequently, even keeping track of what the policy knows is now non-trivial. The analysis of the single-round game becomes more involved, as one needs to argue about ``approximate information-monotonicity''. To handle these issues, our recommendation policy satisfies only a relaxed version of incentive-compatibility.

In the reported types model, we face a similar issue, but achieve a much stronger result: we design a policy which matches our public-types benchmark in the long run. This may seem counterintuitive because ``reported types'' are completely useless to the principal in the single-round game (whereas public types are very useful). Essentially, we reduce the problem to the public types case, at the cost of a much longer exploration.

\xhdr{Discussion.}
This paper, as well as all prior work on incentivizing exploration, relies on very standard yet idealized assumptions of Bayesian rationality and the ``power to commit'' (i.e., principal can announce a policy and commit to implementing it). A recent paper \cite{Jieming-unbiased18} attempts to mitigate these assumptions (in a setting with homogeneous agents). However, some form of the ``power to commit" assumption appears necessary to make any progress.

We do not attempt to \emph{elicit} agents' types when they are not public, in the sense that our recommendation to a given agent is not contingent on anything that this agent reports. However, our result for reported types is already the best possible, in the sense that the explorable set is the same as for public types, so (in the same sense) elicitation is not needed.

\xhdr{Related work.}
Bayesian Exploration with homogenous agents was introduced in \cite{Kremer-JPE14}, and largely resolved: for optimal policy in the case of two actions and deterministic utilities \cite{Kremer-JPE14}, for explorability \cite{ICexplorationGames-ec16}, and for regret minimization and stochastic utilities \cite{ICexploration-ec15}.

Bayesian Exploration with heterogenous agents and public types is studied in \cite{ICexploration-ec15}, under a very strong assumption which ensures explorability of all type-action pairs, and in \cite{ICexplorationGames-ec16}, where a fixed tuple of agent types arrives in each round and plays a game. \cite{ICexplorationGames-ec16} focus on explorability of joint actions. Our approach for the public-type case is similar on a high level, but simpler and more efficient, essentially because we focus on type-action pairs rather than joint actions.

A very recent paper \cite{Kempe-colt18} (ours is independent work) studies incentivizing exploration with heterogenous agents and private types, but allows monetary transfers. Assuming that each action is preferred by some agent type, they design an algorithm with a (very) low regret, and conclude that \emph{diversity helps} in their setting.

Several papers study ``incentivizing exploration" in substantially different models:
with a social network \cite{Bahar-ec16};
with time-discounted utilities \cite{Bimpikis-exploration-ms17};
with monetary incentives \cite{Frazier-ec14,Kempe-colt18};
with a continuous information flow and a continuum of agents \cite{Che-13};
with long-lived agents and ``exploration" separate from payoff generation \cite{Bobby-Glen-ec16,Annie-ec18-traps,Liang-ec18};
with fairness \cite{KKMPRVW17}. Also, seminal papers \cite{Bolton-econometrica99,Keller-econometrica05} study scenarios with
long-lived, exploring agents and no principal.

Recommendation policies with no explicit exploration, and closely related ``greedy algorithm" in multi-armed bandits, have been studied recently \cite{bastani2017exploiting,Sven-aistats18,kannan2018smoothed,externalities-colt18}.
A common theme is that the greedy algorithm performs well under  substantial assumptions on the diversity of types. Yet, it suffers $\Omega(T)$ regret in the worst case.%
\footnote{This is a well-known folklore result in various settings; \eg see \cite{CompetingBandits-itcs18,Sven-aistats18}.}

\OMIT{\cite{Sven-aistats18} consider a ``full-revelation" recommendation system, and show that (under some substantial assumptions) agent heterogeneity leads to exploration.}

Exploration-exploitation tradeoff received much attention over the past decades, usually under the rubric of ``multi-armed bandits", see books \cite{CesaBL-book,Bubeck-survey12,Gittins-book11}. Absent incentives, Bayesian Exploration with public types is a well-studied problem of ``contextual bandits" (with deterministic rewards and a Bayesian prior). A single round of Bayesian Exploration is a version of the Bayesian Persuasion game \cite{Kamenica-aer11}, where the signal observed by the principal is distinct from the state. Exploration-exploitation problems with incentives issues arise in several other scenarios: dynamic pricing, \eg
    \cite{KleinbergL03,BZ09,BwK-focs13},
dynamic auctions \cite{DynAuctions-survey11},
advertising auctions
    \cite{MechMAB-ec09,DevanurK09,Transform-ec10-jacm},
human computation
    \cite{RepeatedPA-ec14,Ghosh-itcs13,Krause-www13},
and repeated actions, \eg
    \cite{Amin-auctions-nips13,Amin-auctions-nips14,Jieming-ec18}.

\section{Model and Preliminaries}
\label{sec:model}

\OMIT{\jmcomment{Reviewer 2: The model description is not clear. In particular, it did not describe what the principal knows about the states, agents' reward functions etc., and what the agents know. Note that it is crucial to be clear about who knows what since the information asymmetry is essential to the model. Is a new state drawn at every round t or is the state of nature the same across all rounds? I think for the model to make sense, it should be the latter case. Then, the questions are: (1) how does the principal's belief about the state of nature evolve? (2) in the benchmark $OPT_{pub}s$ rightly above Section 3, why you need to take expectation over w? Why not just the particular w realized at the beginning of the game? Then at each single round, what do principal and agents know about S?
}}

\emph{Bayesian Exploration} is a game between a principal and $T$ agents. The game consists of $T$ rounds. Each round $t\in [T]$ proceeds as follows: a new agent $t$ arrives, receives a message $m_t$ from the principal, chooses an action $a_t$ from a fixed action space $\A$, and collects a reward $r_t\in [0,1]$ that is immediately observed by the principal. Each agent $t$ has a {\em type} $\theta_t\in\varTheta$, drawn independently from a fixed distribution $\DT$, and an {\em action space} $\mA$ (same for all agents).  There is uncertainty, captured by a ``state of nature" $\omega\in \varOmega$, henceforth simply the \emph{state}, drawn from a Bayesian prior $\prior$ at the beginning of time and fixed across rounds. The {\em reward} $r_t = u(\theta_t,a_t,\omega)\in[0,1]$ of agent $t$ is determined by its type $\theta_t$, the action $a_t\in\mA$ chosen by this agent, and the state $\omega$, for some fixed and deterministic \emph{reward function}
$u:\varTheta\times \A \times \varOmega \to [0,1]$.
The principal's messages $m_t$ are generated according to a randomized online algorithm $\pi$ termed ``recommendation policy".
Thus, an {\em instance}
of Bayesian Exploration consists of the time horizon $T$, the sets $\mA,\varTheta,\varOmega$,
the type distribution $\DT$, the prior $\prior$, and the reward function $u$.

The knowledge structure is as follows. The type distribution $\DT$, the Bayesian prior $\prior$, the reward function $u$, and the recommendation policy are common knowledge. Each agent $t$ knows her own type $\theta_t$, and observes nothing else except the message $m_t$.  We consider three model variants, depending on whether and when the principal learns the agent's type: the type is revealed immediately after the agent arrives (\emph{public types}), the type is revealed only after the principal issues a recommendation (\emph{reported types}), the type is not revealed (\emph{private types}).



Let $H_t$ denote the \emph{history} observed by the principal at round $t$, immediately before it chooses its message $m_t$. Hence, it equals $\{(r_1,\theta_1),\ldots,(r_{t-1},\theta_{t-1}),\theta_t\}$ for public types, $\{(r_1,\theta_1),\ldots,(r_{t-1},\theta_{t-1})\}$ for reported types, and $\{r_1,\ldots,r_{t-1}\}$ for private types.%
\footnote{For randomized policies, the history also contains policy's random seed in each round.}
Formally, this is the input to the recommendation policy in each round $t$. Borrowing terminology from the Bayesian Persuasion literature, we will often refer to the history as the {\em signal}. We denote the set of all possible histories (signals) at time $t$ by $\mH_t$.

\OMIT{A solution to an instance $\I$ of the Bayesian Exploration game is a randomized online algorithm $\pi$ termed ``recommendation policy" which, at each round $t$, maps 
the current history $H_t$ to a distribution over messages $m_t$ which, in general, are arbitrary bit strings of length polynomial in the size of the instance.}

The recommendation policy $\pi$, the type distribution $\DT$, the state distribution $\prior$, and the reward function $u$ induce a joint distribution $\DT(\varOmega,\mH_t)$ over states and histories, henceforth called the {\em signal structure} at round $t$. Note that it is known to agent $t$.

We are ready to state agents' decision model. Each agent $t$, given the realized message $m$, chooses an action $a_t$ so as to maximize her {\em Bayesian-expected reward}
$$\E[r_t]\equiv\E_{(\omega,H_t)\sim\DT(\varOmega,\mH_t)}
\left[\; \E_{m_t\sim\pi(H_t)}[u(\theta_t,a_t,\omega)\;|m_t = m ]\right].$$
Given the instance of Bayesian Exploration, the goal of the principal is to choose a policy $\pi$ that maximizes (Bayesian-expected) {\em total} reward, \ie $\sum_{t=1}^T \E[r_t]$.\footnote{While the principal must commit to the policy given only the problem instance, the policy itself observes the history and thus can adapt recommendations to inferences about the state based on the history.  See Example~\ref{exp:simple}.}

We assume that the sets $\mA$, $\varTheta$ and $\varOmega$ are finite. We use $\omega_0$ as the random variable for the state, and write $\Pr[\omega]$ for $\Pr[\omega_0=\omega]$. Similarly, we write $\Pr[\theta]$ for $\Pr[\theta_t=\theta]$.

\xhdr{Bayesian-incentive compatibility.}
For public types, we assume the message $m_t$ in each round is a recommended action $a \in \mA$ which, for convenience, we sometimes write as $m_t(\theta_t)$. For private and reported types, we assume that the message $m_t$ in each round is a \emph{menu} mapping types to actions, i.e., $m_t:\varTheta\rightarrow\mA$.
 We further assume $\pi$ is Bayesian incentive-compatible.
\begin{definition}
Let $\EE_t$ be the event that the agents have followed principal's recommendations before round $t$, \ie $a_s = m_s(\theta_s)$ for all rounds $s<t$.
The recommendation policy $\pi$ is {\em Bayesian incentive compatible} (\emph{BIC}) if for all rounds $t$ and messages $m$ such that
\[ \Pr_{(\omega,H_t)\sim\DT(\varOmega,\mH_t)}
    [m = \pi(H_t)\; \mid\; \EE_{t}] > 0,
\]
it holds that for all types $\theta$ and actions $a$,
\begin{align}\label{eq:model-BIC}
\E\left[\; u(\theta,m(\theta),\omega) - u(\theta,a,\omega) \; \mid\; m_t=m, \EE_{t}\;\right] \geq 0,
\end{align}
where the expectation is over $(\omega,H_t)\sim\DT(\varOmega,\mH_t)$.
\end{definition}

\noindent The above assumptions are without loss of generality, by a suitable version of Myerson's ``revelation principle".

\xhdr{Explorability and benchmarks.}
For public types, a type-action pair $(\theta,a)\in \Theta\times \A$ is called \emph{eventually-explorable} in state $\omega$ if there is some BIC recommendation policy that, for $T$ large enough, eventually recommends this action to this agent type with positive probability. Then action $a$ is called \emph{eventually-explorable} for type $\theta$ and state $\omega$. The set of all such actions is denoted $\AExp_{\omega,\theta}$.

Likewise, for private types, a menu is called \emph{eventually-explorable} in state $\omega$ if there is some BIC recommendation policy that eventually recommends this menu with positive probability. The set of all such menus is denoted $\MExp_{\omega}$.

Our benchmark is the best eventually-explorable recommendation for each type. For public and private types, resp., this is
\begin{align}
\OPTpub &= \sum_{\theta \in \varTheta, \omega\in \varOmega} \Pr[\omega] \cdot \Pr[\theta] \cdot \max_{a \in \AExp_{\omega,\theta}} u(\theta, a, \omega).
    \label{eq:bench-public}\\
\OPTpri &= \sum_{\omega\in \varOmega} \Pr[\omega] \cdot\max_{m \in \MExp_{\omega}}\sum_{\theta \in \varTheta} \Pr[\theta] \cdot  u(\theta, m(\theta), \omega).\label{eq:bench-private}
\end{align}
We have $\OPTpub \geq \OPTpri$, essentially because any BIC policy for private types can be simulated as a BIC policy for public types. We provide an example (Example \ref{exp:simple}) when $\OPTpub > \OPTpri$.

\OMIT{
Note that, for all settings, no BIC recommendation policy can out-perform the corresponding benchmark.  Our main technical contributions are (computationally efficient) policies that get arbitrarily close to these benchmarks as the number of agents grows.}


\section{Comparative Statics}
\label{sec:statics}

\newcommand{\pairs}{\AExp_{\omega}}
\newcommand{\pairsPub}{\pairs^{\term{pub}}}
\newcommand{\pairsPri}{\pairs^{\term{pri}}}

\newcommand{\support}{\term{support}}

We discuss how the set of all eventually-explorable type-action pairs (\emph{explorable set}) is affected by the model choice and the diversity of types. The explorable set is all information that can possibly be learned in the public-types model. All else equal, settings with larger explorable set have greater or equal total expected reward, both in benchmark \eqref{eq:bench-public} and in our approximation guarantees. For private types, the exploration set provides an ``upper bound" on the information available to the principal, because the principal does not directly observe the agent types.
\OMIT{Our first result shows that models with public or reported types can explore (strictly) more actions for each type than models with private types.  Thus more information about types (strictly) improves the outcomes. Our second result shows that greater diversity (in the sense of a greater number of possible agent types) improves exploration for public or reported types but, in fact, can {\em harm} exploration for private types.  The intuition is that with private types, diversity can muddle the information of the principal, hindering her ability to learn about the state, whereas for public or reported types diversity only helps the principal refine her beliefs about the state.
} 

\xhdr{Explorability and the model choice.}
Fix an instance of Bayesian Exploration. Let $\pairsPub$ and $\pairsPri$ be the explorable set for a given state $\omega$, for public and private types, respectively.%
\footnote{Equivalently, $\pairsPri$ is the set of all type-action pairs $(\theta,m(\theta))$ that appear in some eventually-explorable menu $m\in \MExp_{\omega}$ in state $\omega$ with private types.}
We will show in Section~\ref{sec:reported} that the explorable set for reported types is $\pairsPub$, too.


\begin{claim}
$\pairsPri \subseteq \pairsPub$.
\end{claim}

The idea is that one can simulate any BIC recommendation policy for private types with a BIC recommendation policy for public types; we omit the details.


Interestingly, $\pairsPri$ can in fact be a {\em strict} subset of $\pairsPub$:

\begin{example}
	\label{exp:simple}
	There are 2 states, 2 types and 2 actions:
	$\varOmega = \varTheta = \A = \{0,1\}$.
	States and types are drawn uniformly at random:
	$\Pr[\omega =0] =\Pr[\theta =0] = \tfrac12$.
	Rewards are defined as follows:\\
	\begin{table}[H]
		\centering
		\begin{tabular}{|c||c|c|}
			\hline
			&$a=0$&$a=1$\\
			\hline
			\hline
			$\theta = 0$& $u = 3$ & $u =4$\\
			\hline
			$\theta = 1$& $u = 2$ & $u =0$\\
			\hline
		\end{tabular}
		\quad
		\begin{tabular}{|c||c|c|}
			\hline
			&$a=0$&$a=1$\\
			\hline
			\hline
			$\theta = 0$& $u = 2$ & $u =0$\\
			\hline
			$\theta = 1$& $u = 3$ & $u =4$\\
			\hline
		\end{tabular}
		\caption{Rewards $u(\theta,a,\omega)$ when $\omega =0 $ and $\omega = 1$.}
	\end{table}
\end{example}

\begin{claim}
	In Example \ref{exp:simple}, $\pairsPri$ is a strict subset of $\pairsPub$.
\end{claim}

\begin{proof}
Action 0 is preferred by both types initially. Thus in the first round, the principal must recommend action $0$ in order for the policy to be BIC.  Hence type-action pairs $\{(0,0),(1,0)\}$ are eventually-explorable in all models.

In the second round, the principal knows the reward of the first-round agent.  When types are public or reported, the reward together with the type is sufficient information for the principal to learn the state.  Moving forward, the principal can now recommend the higher-reward action for each type (either directly or, in the case of reported types, through a menu).  Thus, type-action pair $(0,1)$ is eventually-explorable when $\omega=0$ and, similarly, type-action pair $(1,1)$ is eventually-explorable when $\omega=1$.

For private types, samples from the first-round menu (which, as argued above, must recommend action $0$ for both types) do not convey any information about the state, as they have the same distribution in both states. Therefore, action $1$ is not eventually-explorable, for either type and either state.
\end{proof}

\xhdr{Explorability and diversity of agent types.}
Fix an instance of Bayesian Exploration with type distribution $\DT$. We consider how the explorable set changes if we modify the type distribution $\DT$ in this instance to some other distribution $\DT'$. Let $\pairs$ and $\pairs'$ be the corresponding explorable sets, for each state $\omega$.

For public and reported types, we show that the explorable set is determined by the support set of $\DT$, denoted $\support(\DT)$, and can only increase if the support set increases:

\begin{claim}\label{cl:statics-diversity-public}
Consider Bayesian Exploration with public types. Then:
\begin{OneLiners}
\item[(a)] if $\support(\DT)=\support(\DT')$ then $\pairs=\pairs'$.
\item[(b)] if $\support(\DT)\subset \support(\DT')$ then $\pairs\subseteq \pairs'$.
\end{OneLiners}
\end{claim}

\OMIT{ 
\begin{claim}\label{cl:statics-diversity-public}
Consider Bayesian exploration with public or reported types. Then:
\begin{OneLiners}
\item[(a)] if the supports of distributions $\DT$ and $\DT'$ are the same, then $\pairs=\pairs'$.
\item[(b)] if the support of distribution $\DT$ is contained in the support of distribution $\DT'$ then $\pairs\subseteq \pairs'$.
\end{OneLiners}
\end{claim}
} 

\begin{proof}[Proof Sketch]
Consider public types (the case of reported types then follows by arguments in Section~\ref{sec:reported}).  Let $\pi$ be a BIC recommendation policy for the instance with type distribution $\DT$ and suppose $\pi$ eventually explores type-action pairs $\pairs$ for this instance and state $\omega$.  Consider the instance with type distribution $\DT'$.  Extend $\pi$ to a policy $\pi'$ as follows: let $T'$ be the subsequence of $T$ for which $\DT(\theta_t)>0$. If $t\not\in T'$, then recommend the action $a$ that maximizes agent $t$'s Bayesian-expected reward.  If $t\in T'$, then consider the sub-history $H\equiv H^{T'}_t$ restricted to $T'$ and recommend action $a\sim\pi(H)$.  Then $\pi'$ is BIC for the instance with type distribution $\DT'$. Furthermore, $\pi'$ eventually explores the same set of type-action pairs $\pairs$ for this modified instance as well (and possibly more) as every history that occurs with positive probability in the original instance occurs as a sub-history in the modified instance with positive probability as well.
\end{proof}

For private types, the situation is more complicated. More types can help for some problem instances. For example, if different types have disjoint sets of available actions (more formally: say, disjoint sets of actions with positive rewards) then we are essentially back to the case of reported types, and the conclusions in Claim~\ref{cl:statics-diversity-public} apply. On the other hand, we can use Example \ref{exp:simple} to show that more types can hurt explorability when types are private. Recall that in this example, for private types only action 0 can be recommended. Now consider a less diverse instance in which only type 0 appears. After one agent in that type chooses action 0, the state is revealed to the principal. For example, when the state $\omega = 0$, action $1$ can be recommended to future agents. This shows that,  in this example, explorable set increases when we have fewer types.


\section{Public Types}
\label{sec:public}

In this section, we develop our recommendation policy for public types. Throughout, $\OPT = \OPTpub$.

\begin{theorem}
\label{thm:public}
Consider an arbitrary instance of Bayesian Exploration with public types.
There exists a BIC recommendation policy with expected total reward at least $\left(T - C \right) \cdot \OPT$, for some constant $C$ that depends on the problem instance but not on $T$. This policy explores all type-action pairs that are eventually-explorable for a given state.
\end{theorem}

\subsection{A single round of Bayesian Exploration}
\label{sec:public_single}

\xhdr{Signal and explorability.}
We first analyze what actions can be explored by a BIC policy in a single round $t$ of Bayesian Exploration for public types, as a function of the history. Throughout, we suppress $\theta$ and $t$ from our notation.
Let $S$ be a random variable equal to the history at round $t$ (referred to as a {\em signal} throughout this section), $s$ be a realization of $S$, and $\S=\DT(\Omega,\mH)$ be the signal structure: the joint distribution of $(\omega,S)$.  Note different policies induce different histories and hence different signal structures.  Thus it will be important to be explicit about the signal structure throughout this section.

%


\begin{definition}
	Consider a single-round of Bayesian Exploration when the principal receives signal $S$ with signal structure $\S$. An action $a \in \A$ is called {\em signal-explorable for a realized signal $s$} if there exists a BIC recommendation policy $\pi$ such that $\Pr[\pi(s) = a] > 0$. The set of all such actions is denoted as $\EX_s[\S]$. The {\em signal-explorable set}, denoted $\EX[\S]$, is the random subset of actions $\EX_S[\S]$.
\end{definition}



\xhdr{Information-monotonicity.}
We compare the information content of two signals using the notion of conditional mutual information (see
Appendix~\ref{app:info-theory} for background). Essentially, we show that a more informative signal leads to the same or larger explorable set.

\begin{definition}
	We say that signal $S$ \emph{is at least as informative} as signal $S'$ if $I(S' ; \omega\mid S) = 0$.
\end{definition}

Intuitively, the condition $I(S';\omega_0|S)= 0$  means if one is given random variable $S$, one can learn no further information from $S'$ about $\omega_0$. Note that this condition depends not only on the signal structures of the two signals, but also on their joint distribution.

\begin{lemma}
	\label{lem:infomono}
	Let $S,S'$ be two signals with signal structures $\S,\S'$. If $S$ is at least as informative as $S'$, then $\EX_{s'}[\S'] \subseteq \EX_s[\S]$ for all $s' ,s$ such that $\Pr[S= s, S'= s'] > 0$.
\end{lemma}

\begin{proof}
Consider any BIC recommendation policy $\pi'$ for signal structure $\S'$. We construct $\pi$ for signal structure $\S$ by setting $\Pr[\pi(s) = a] = \sum_{s'} \Pr[\pi'(s') = a] \cdot Pr[S' = s'\mid S = s]$. Notice that $I(S' ; \omega_0\mid S) = 0$ implies $S'$ and $\omega_0$ are independent given $S$, i.e $\Pr[S' = s'\mid S=s] \cdot \Pr[\omega_0 = \omega\mid S=s] = \Pr[S'=s', \omega_0 = \omega\mid S=s]$ for all $s,s',\omega$. Therefore, for all $s'$ and $\omega$,
\begin{align*}
& \textstyle \sum_s \Pr[S' = s'\mid S = s] \cdot \Pr[\omega_0 = \omega, S= s]\\
&\qquad=\textstyle \sum_s \Pr[S' = s'\mid S=s]
    \cdot \Pr[\omega_0 = \omega\mid S=s] \cdot \Pr[S=s] \\
&\qquad= \textstyle  \sum_s \Pr[S'=s',\omega_0 =\omega\mid S=s] \cdot \Pr[S=s] \\
&\qquad=\textstyle  \sum_s \Pr[S=s,S'=s',\omega_0 =\omega] \\
&\qquad=\Pr[\omega_0 =\omega, S'=s'].
\end{align*}

Therefore $\pi'$ being BIC implies that $\pi$ is also BIC. Indeed, for any $a,a' \in \A$ and $\theta \in \varTheta$, by plugging in the definition of $\pi$,
\begin{align*}
&\textstyle \sum_{\omega,s}\; \Pr[\omega_0 = \omega, S = s] \cdot (u(\theta,a', \omega) - u(\theta,a,\omega)) \cdot \Pr[\pi(s) = a] \\
&\;= \textstyle \sum_{\omega,s'}\;\Pr[\omega_0 = \omega, S' = s'] \cdot (u(\theta,a', \omega) - u(\theta,a,\omega)) \cdot \Pr[\pi'(s') = a]\\
&\;\geq 0.
\end{align*}


Finally, for any $s', s ,a$ such that $Pr[S' = s',S = s] >0 $ and $\Pr[\pi'(s') = a] >0$, we have $\Pr[\pi(s) = a] > 0$. This implies $\EX_{s'}[\S'] \subseteq \EX_s[\S]$.
\end{proof}

\xhdr{Max-Support Policy.}
We can solve the following LP to check whether a particular action $a_0 \in\A$ is signal-explorable given a particular realized signal $s_0\in\X$. In this LP, we represent a policy $\pi$ as a set of numbers
    $x_{a,s} = \Pr[\pi(s)=a]$,
for each action $a\in \A$ and each feasible signal $s\in \X$.

\begin{figure}[H]
\begin{mdframed}
\vspace{-3mm}
\begin{alignat*}{2}
&\textbf{maximize }    x_{a_0,s_0}\  \\
&\textbf{subject to: }\\
    & \textstyle \sum_{\omega \in \varOmega, s \in \X} \;
    \Pr[\omega] \cdot \Pr[s \mid  \omega] \cdot\\
        &\left(u(\theta, a, \omega) - u(\theta, a', \omega)\right) \cdot x_{a,s} \geq 0   &\qquad & \forall a,a' \in \A \\
    & \textstyle \sum_{a\in \A}\; x_{a,s} = 1,  \ &\ & \forall s \in \X \\
    & x_{a,s} \geq 0,  \ &\ & \forall s \in \X, a\in \A
\end{alignat*}
\end{mdframed}
\label{fig:public_lp}
\end{figure}

Since the constraints in this LP characterize any BIC recommendation policy, it follows that action $a_0$ is signal-explorable given realized signal $s_0$ if and only if the LP has a positive solution. If such solution exists, define recommendation policy $\pi = \pi^{a_0,s_0}$ by setting
    $\Pr[\pi(s) = a] = x_{a,s}$ for all $a\in \A, s\in \X$.
Then this is a BIC recommendation policy such that
    $\Pr[\pi(s_0) = a_0] > 0$.

\begin{definition}
Given a signal structure $\S$, a BIC recommendation policy $\pi$ is called  \emph{max-support} if $\forall s \in \X$  and signal-explorable action $a\in \A$ given $s$, $\Pr[\pi(s) = a] > 0$.
\end{definition}

It is easy to see that we obtain max-support recommendation policy by averaging the $\pi^{a,s}$ policies defined above. Specifically, the following policy is BIC and max-support:
\begin{align}\label{eq:pimax}
\pi^{\max} = \frac{1}{|\X|} \sum_{s \in \X} \frac{1}{|\EX_s[\S]|} \sum_{a \in \EX_s[\S]} \pi^{a,s}.
\end{align}

%

\xhdr{Maximal Exploration.}
We design a subroutine \term{MaxExplore} which outputs a sequence of actions with two properties: it includes every signal-explorable action at least once, and each action in the sequence marginally distributed as $\pi^{\max}$. The length of this sequence, denoted $L_{\theta}$, should satisfy
\begin{align}\label{eq:public-L}
L_{\theta} \geq \max_{(a,s)\in \A\times \X \text{ with } \Pr[\pi^{\max}(s)=a] \neq 0} \quad \frac{1}{\Pr[\pi^{\max}(s)=a]}.
\end{align}

This step is essentially from \cite{ICexplorationGames-ec16};
\OMIT{\jmcomment{Reviewer 2 mentioned that we cite your EC paper twice. I know Alex wants to show that this is from your working paper which is not the version for EC.  But it might confuse the reviewer. Not sure if we should do this.}}
 we provide the details below for the sake of completeness.
The idea is to put $C_a = L_{\theta} \cdot \Pr[\pi^{\max}(S) = a]$ copies of each action $a$ into a sequence of length $L_{\theta}$ and randomly permute the sequence.
 However, $C_a$ might not be an integer, and in particular may be smaller than 1. The latter issue is resolved by making $L_{\theta}$ sufficiently large. For the former issue, we first put $\lfloor C_a \rfloor$ copies of each action $a$ into the sequence, and then sample the remaining
    $L_\theta - \sum_a \lfloor C_a \rfloor$
actions according to distribution
    $\pRes(a) = \frac{C_a - \lfloor C_a \rfloor}{L_\theta - \sum_a \lfloor C_a \rfloor}$.
For details, see Algorithm \ref{alg:public_explore}.
 \begin{algorithm}[H]
    \caption{Subroutine MaxExplore}
    	\label{alg:public_explore}
    \begin{algorithmic}[1]
	\STATE \textbf{Input:} type $\theta$, signal $S$ and signal structure $\S$.
	\STATE \textbf{Output:} a list of actions $\alpha$
	\STATE Compute $\pi^{\max}$ as per \eqref{eq:pimax}
		\STATE Initialize $Res = L_{\theta}$.
		\FOR {each action $ a \in \A$}
							\STATE $C_a \leftarrow  L_{\theta} \cdot \Pr[\pi^{\max}(S) = a]$
                     		\STATE Add $\lfloor C_a \rfloor$ copies of action $a$ into list $\alpha$.
			\STATE $Res \leftarrow Res -\lfloor C_a \rfloor $.
			\STATE $\pRes(a)\leftarrow  C_a -  \lfloor C_a\rfloor$
		\ENDFOR
		\STATE $\pRes(a) \leftarrow \pRes(a) / Res$, $\forall a \in \A$.
		\STATE Sample $Res$ many actions from distribution $\pRes$ independently and add these actions into $\alpha$.
		\STATE Randomly permute the actions in $\alpha$.
	\RETURN $\alpha$.	
     \end{algorithmic}
\end{algorithm}

\begin{claim}
\label{clm:maxexplore}
Given type $\theta$ and signal $S$, MaxExplore outputs a sequence of $L_{\theta}$ actions.
Each action in the sequence marginally distributed as $\pi^{\max}$.
For any action $a$ such that $\Pr[\pi^{\max} =a] >0$, $a$ shows up in the sequence at least once with probability exactly 1.
MaxExplore runs in time polynomial in $L_{\theta}$, $|\A|$, $|\varOmega|$ and $|\X|$ (size of the support of the signal).
\end{claim}


\subsection{Main Recommendation Policy}
\label{sec:public_main}

Algorithm \ref{alg:public_main} is the main procedure of our recommendation policy. It consists of two parts: \emph{exploration}, which explores all the eventually-explorable actions, and \emph{exploitation}, which simply recommends the best explored action for a given type. The exploration part proceeds in phases. In each phase $\ell$, each type $\theta$ gets a sequence of $L_{\theta}$ actions from MaxExplore using the data collected before this phase starts. The phase ends when every agent type $\theta$ has finished $L_{\theta}$ rounds. We pick parameter $L_{\theta}$ large enough so that the condition
\eqref{eq:public-L} is satisfied for all phases $\ell$ and all possible signals $S=S_\ell$. (Note that $L_{\theta}$ is finite because there are only finitely many such signals.)
  After $|\A| \cdot | \varTheta|$ phases, our recommendation policy enters the exploitation part. See Algorithm \ref{alg:public_main} for  details.

 \begin{algorithm}[t]
    \caption{Main procedure for public types }
    	\label{alg:public_main}
    \begin{algorithmic}[1]
    \STATE Initialization: signal $S_1 = \S_1 = \perp$,
             phase count $\ell = 1$, index $i_{\theta} = 0$ for each type $\theta \in \varTheta$.
	\FOR {rounds $t=1$ to $T$}
		\IF {$\ell \leq |\A|\cdot |\varTheta|$}
		 \STATE \COMMENT{Exploration}
		\STATE Call thread $\thread(\theta_t)$.
			\IF {every type $\theta$ has finished $L_{\theta}$ rounds in the current phase ($i_{\theta} \geq L_{\theta}$)}
				\STATE Start a new phase: $\ell \leftarrow \ell + 1$.
				\STATE Let $S_\ell$ be the signal for phase $\ell$: 
                 the set of all observed type-action-reward triples.
        \STATE Let $\S_\ell$ be the signal structure for $S_\ell$
         given the realized type sequence $(\theta_1,...,\theta_t)$.
			\ENDIF
		\ELSE
			\STATE \COMMENT{Exploitation}
			\STATE Recommend the best explored action for agent type $\theta_t$.
		\ENDIF
	\ENDFOR
     \end{algorithmic}
\end{algorithm}

There is a separate thread for each type $\theta$, denoted $\thread(\theta)$,  which is called whenever an agent of this type shows up; see Algorithm \ref{alg:public_sub}. In a given phase $\ell$, it recommends the $L_{\theta}$ actions computed by MaxExplore, then switches to the best explored action. The thread only uses the information collected before the current phase starts: the signal $S_\ell$ and signal structure $\S_\ell$.

 \begin{algorithm}[h]
    \caption{Thread for agent type $\theta$: $\thread(\theta)$ }
    	\label{alg:public_sub}
    \begin{algorithmic}[1]
		\IF {this is the first call of $\thread(\theta)$ of the current phase}
			\STATE Compute a list of $L_{\theta}$ actions $\alpha_{\theta} \leftarrow $ MaxExplore($\theta, S_\ell, \S_\ell$).
			\STATE Initialize the index of type $\theta$: $i_{\theta} \leftarrow 0$.
		\ENDIF
		\STATE $i_{\theta} \leftarrow i_{\theta} + 1$.
		\IF {$i_{\theta} \leq L_{\theta}$}
			\STATE Recommend action $\alpha_{\theta} [i_{\theta}]$.
		\ELSE
			\STATE Recommend the best explored action of type $\theta$.
		\ENDIF
     \end{algorithmic}
\end{algorithm}

The BIC property follows easily from Claim \ref{clm:maxexplore}. The key is that Algorithm \ref{alg:public_main} explores all  eventually-explorable type-action pairs.



\OMIT{The performance analysis proceeds as follows. First, we upper-bound the expected number of rounds of a phase (Lemma \ref{lem:epoch}). Then we show, in Lemma \ref{lem:exp_public}, that Algorithm \ref{alg:public_main} explores all  eventually-explorable type-action pairs in $|\A| \cdot |\varTheta|$ phases. We use these two lemmas to prove the main theorem.}

\OMIT{ 
\begin{lemma}
\label{lem:epoch}
The expected number of rounds in each phase $\ell$ at most
$ \sum_{\theta\in\varTheta} \frac{L_{\theta}}{\Pr[\theta]}$.
\end{lemma}

\begin{proof}
Phase ends as soon as each type has shown up at least $L_{\theta}$ times. The expected number of rounds by which this happens is at most
$ \sum_{\theta\in\varTheta} \frac{L_{\theta}}{\Pr[\theta]}$.
\end{proof}
} 

\OMIT{
Notice that in Algorithm \ref{alg:public_main} the partition of phases depends only on realized types $\theta_1,...,\theta_T$.
\begin{claim}
Given the sequence $\theta_1,...,\theta_T$, the partition of phases in Algorithm \ref{alg:public_main} is fixed.
\end{claim}
} 

The following lemma compares the exploration of Algorithm \ref{alg:public_main} with $l$ phases and some other BIC recommendation policy with $l$ rounds. Notice that a phase in Algorithm \ref{alg:public_main} has many rounds.
\begin{lemma}
\label{lem:exp_public}
Fix phase $\ell>0$ and the sequence of agent types $\theta_1,...,\theta_T$. Assume Algorithm \ref{alg:public_main} has been running for at least $\min(l, |\A|\cdot |\varTheta|)$ phases.
For a given state $\omega$, if type-action pair $(\theta,a)$ can be explored by some BIC recommendation policy $\pi$ at round $\ell$ with positive probability, then such action is explored by Algorithm \ref{alg:public_main} by the end of phase $\min(l, |\A|\cdot |\varTheta|)$ with probability $1$.
\end{lemma}

\begin{proof}
We prove this by induction on $\ell$ for $\ell \leq |\A|\cdot |\varTheta|$. Base case $\ell=1$ is trivial by Claim \ref{clm:maxexplore}. Assuming the lemma is correct for $\ell-1$, let's prove it's correct for $\ell$.

Let $S= S_l$ be the signal of Algorithm \ref{alg:public_main} by the end of phase $\ell-1$.  Let $S'$ be the history of $\pi$ in the first $\ell-1$ rounds. More precisely,
    $S' = (R, H_1,...,H_{l-1})$,
where $R$ is the internal randomness of policy $\pi$, and
    $H_t = (\Theta_t, A_t, u(\Theta_t, A_t, \omega_0))$
is the type-action-reward triple in round $t$ of policy $\pi$.

The proof plan is as follows. We first show that $I(S';\omega_0|S) =0 $. Informally, this means the information collected in the first $l-1$ phases of Algorithm \ref{alg:public_main} contains all the information $S'$ has about the state $w_0$. After that, we will use the information monotonicity lemma to show that phase $l$ of Algorithm \ref{alg:public_main} explores all the action-type pairs $\pi$ might explore in round $l$.

First of all, we have
\begin{align*}
I(S'; \omega_0\mid  S)
    &= I(R,H_1,...,H_{l-1}; \omega_0\mid  S)\\
   & = I(R; \omega_0\mid  S) + I(H_1,...,H_{l-1}; \omega_0\mid S, R) \\
    &= I(H_1,...,H_{l-1}; \omega_0\mid S, R).
\end{align*}

By the chain rule of mutual information, we have
\begin{align*}
 I(H_1,...,H_{l-1}; \omega_0\mid S, R) 
 = I(H_1;\omega_0\mid S,R) + \cdots + I(H_{l-1}; \omega_0\mid S,R,H_1,...,H_{l-2}).
\end{align*}

For all $t \in [l-1]$, we have
\begin{align*}
I(H_t; \omega_0\mid S,R,H_1,...,H_{t-1}) 
&= I(\Theta_t, A_t, u(\Theta_t, A_t, \omega_0); \omega_0\mid S,R,H_1,...,H_{t-1}) \\
&= I(\Theta_t ; \omega_0\mid S,R,H_1,...,H_{t-1})\\
&\qquad+  I(A_t, u(\Theta_t, A_t, \omega_0); \omega_0\mid S,R,H_1,...,H_{t-1},\Theta_t) \\
&= I(A_t, u(\Theta_t, A_t, \omega_0); \omega_0\mid S,R,H_1,...,H_{t-1},\Theta_t).
\end{align*}
Notice that the suggested action $A_t$ is a deterministic function of randomness of the recommendation policy $R$,  history of previous rounds $H_1,...,H_{t-1}$ and type in the current round $\Theta_t$. Also notice that, by induction hypothesis, $u(\Theta_t, A_t, \omega_0)$ is a deterministic function of $S,R,H_1,...,H_{t-1},\Theta_t, A_t$. Therefore we have
\[
I(H_t; \omega_0\mid S,R,H_1,...,H_{t-1}) = 0, \qquad \forall t \in [l-1].
\]
Then we get
$ I(S'; \omega_0 \mid  S) = 0.$

By Lemma \ref{lem:infomono}, we know that $\EX[\S'] \subseteq \EX[\S]$. For state $\omega$, there exists a signal $s'$ such that $\Pr[S'=s'\mid \omega_0 =\omega] >0 $ and $a \in \EX_{s'} [\S']$. Now let $s$ be the realized value of $S$ given $\omega_0 = \omega$, we know that $\Pr[S'=s'\mid S=s] >0$, so $a \in \EX_s[\S]$. By Claim \ref{clm:maxexplore}, we know that at least one agent of type $\theta$ in phase $\ell$ of Algorithm \ref{alg:public_main} will choose action $a$.

Now consider the case when $\ell > |\A| \cdot |\varTheta|$. Define \ALG to be the variant of Algorithm \ref{alg:public_main} such that it only does exploration (removing the if-condition and exploitation in Algorithm \ref{alg:public_main}). For $\ell > |\A| \cdot |\varTheta|$, the above induction proof still work for \ALG, i.e. for a given state $\omega$, if an action $a$ of type $\theta$ can be explored by a BIC recommendation policy $\pi$ at round $\ell$, then such action is guaranteed to be explored by \ALG by the end of phase $\ell$. Now we are going to argue that \ALG won't explore any new action-type pairs after phase $|\A| \cdot |\varTheta|$. Call a phase exploring if in that phase \ALG explores at least one new action-type pair. As there are  $ |\A| \cdot |\varTheta|$ type-action pairs, \ALG can have at most $ |\A| \cdot |\varTheta|$ exploring phases. On the other hand, once \ALG has a phase that is not exploring, because the signal stays the same after that phase, all phases afterwards are not exploring. So, \ALG does not have any exploring phases after phase $|\A| \cdot |\varTheta|$. For $\ell > |\A| \cdot |\varTheta|$, the first $|   \A| \cdot |\varTheta|$ phases of Algorithm \ref{alg:public_main} explores the same set of type-action pairs as the first $\ell$ phases of \ALG.
\end{proof}

\begin{proof}[Proof of Theorem \ref{thm:public}]
Algorithm \ref{alg:public_main} is BIC  by Claim \ref{clm:maxexplore}. By Lemma \ref{lem:exp_public}, Algorithm \ref{alg:public_main} explores all the eventually-explorable type-actions pairs after $|\A|\cdot |\varTheta|$ phases.
After that, for each agent type $\theta$, Algorithm \ref{alg:public_main} recommends the best explored action: \\$ \arg\max_{a \in \AExp_{\omega,\theta}} u(\theta, a, \omega)$ with probability exactly 1.%
\footnote{This holds with probability exactly 1, provided that our algorithm finishes $|\A|\cdot |\varTheta|$  phases. If some undesirable low-probability event happens, \eg if all agents seen so far have had the same type, our algorithm would never finish $|\A|\cdot |\varTheta|$ phases.}

Therefore Algorithm \ref{alg:public_main} gets reward $\OPT$ except rounds in the first $|\A|\cdot |\varTheta|$ phases.  It remains to prove that the expected number of rounds in exploration (i.e. first $|\A|\cdot |\varTheta|$ phases) does not depend on the time horizon $T$. Let $N_\ell$ be the duration of phase $\ell$.
Recall that the phase ends as soon as each type has shown up at least $L_{\theta}$ times. It follows that
$ \E[N_\ell] \leq  \sum_{\theta\in\varTheta} \frac{L_{\theta}}{\Pr[\theta]}$.
So, one can take $C = |\A|\cdot |\varTheta|\cdot \sum_{\theta\in\varTheta} \frac{L_{\theta}}{\Pr[\theta]}$.
\end{proof}



\subsection{Extension to Reported Types}
\label{sec:reported}

\newcommand{\pipub}{\pi_{\term{pub}}}

We sketch how to extend our ideas for public types to handle the case of reported types. We'd like to simulate the recommendation policy for public types, call it $\pipub$. We simulate it separately for the exploration part and the exploitation part. The exploitation part is fairly easy: we provide a menu that recommends the best explored action for each agent types.

In the exploration part, in each round $t$ we guess the agent type to be $\hat{\theta}_t$, with equal probability among all types.%
\footnote{We guess the types uniformly, rather than according to their probabilities, because our goal is to explore each type for certain number of rounds. Guessing a type according to its probability will only make rare types appear even rarer.}
The idea is to simulate $\pipub$ only in \emph{lucky rounds} when we guess correctly, \ie $\hat{\theta}_t=\theta_t$. Thus,
in each round $t$ we simulate the $\ell_t$-th round of $\pipub$, where $\ell_t$ is the number of lucky rounds before round $t$. In each round $t$ of exploration, we suggest the following menu. For type $\hat{\theta}_t$, we recommend the same action as $\pipub$ would recommend for this type in the $\ell_t$-th round, namely
    $\hat{a}_t = \pipub^{\ell_t}(\hat{\theta}_t)$.
For any other type, we recommend the action which has the best expected reward given the ``common knowledge" (information available before round $1$) and the action $\hat{a}_t$. This is to ensure that in a lucky round, the menu does not convey any information beyond action  $\hat{a}_t$. When we receive the reported type, we can check whether our guess was correct. If so, we input the type-action-reward triple back to $\pipub$. Else, we ignore this round, as if it never happened.

Thus, our recommendation policy eventually explores the same type-action pairs as $\pipub$. The expected number of rounds increases by the factor of $|\varTheta|$. Thus, we have the following theorem.

\begin{theorem}
\label{thm:reported}
Consider Bayesian Exploration with reported types.
There exists a BIC recommendation policy whose expected total reward is at least $\left(T - C \right) \cdot \OPTpub$,
for some constant $C$ that depends on the problem instance but not on $T$.
This policy explores all type-action pairs that are eventually-explorable
for public types.
\end{theorem}



\section{Private Types}
\label{sec:private_nc}

Our recommendation policy for private types satisfies a relaxed version of the BIC property, called \emph{$\delta$-BIC}, where the right-hand side in \eqref{eq:model-BIC} is $-\delta$ for some fixed $\delta>0$. We assume a more permissive behavioral model in which agents obey such policy.

The main result is as follows. (Throughout, $\OPT = \OPTpri$.)

\begin{theorem}
\label{thm:private_nocc}
Consider Bayesian Exploration with private types, and fix $\delta > 0$. There exists a $\delta$-BIC recommendation policy with expected total reward at least $\left(T - C \log T \right) \cdot \OPT$, where $C$ depends on the problem instance but not on time horizon $T$.
\end{theorem}


The recommendation policy proceeds in phases: in each phase, it explores all menus that can be explored given the information collected so far. The crucial step in the proof is to show that:

\begin{property}
\item the first $l$ phases of our recommendation policy explore all the menus that could be possibly explored by the first $l$ rounds of any BIC recommendation policy.
    \label{prop:private-exploration}
\end{property}

The new difficulty for private types comes from the fact that we are exploring menus instead of type-actions pairs, and we do not learn the reward of a particular type-action pair immediately. This is because a recommended menu may map several different types to the chosen action, so knowing the latter does not immediately reveal the agent's type. Moreover, the full ``outcome" of a particular menu is a distribution over action-reward pairs, it is, in general, impossible to learn this outcome exactly in any finite number of rounds.  Because of these issues, we cannot obtain Property \refprop{prop:private-exploration} exactly. Instead, we achieve an approximate version of this property, as long as we explore each menu enough times in each phase.

We then show that this approximate version of \refprop{prop:private-exploration}  suffices to guarantee explorability, if we relax the incentives property of our policy from BIC to $\delta$-BIC, for any fixed $\delta>0$. In particular, we prove an approximate version of the information-monotonicity lemma (Lemma~\ref{lem:infomono}) which (given the approximate version of \refprop{prop:private-exploration}) ensures that our recommendation policy can explore all the menus that could be possibly explored by the first $l$ rounds of any BIC recommendation policy.

\subsection{A Single round of Bayesian Exploration}
\label{sec:private_single}

Recall that for a random variable $S$, called \emph{signal}, the signal structure is a joint distribution of $(\omega,S)$.

\begin{definition}
Consider a single-round of Bayesian Exploration when the principal has signal $S$ from signal structure $\S$. For any $\delta \geq 0$, a menu $m \in \M$ is called $\delta$-signal-explorable, for a given signal $s$, if there exists a single-round $\delta$-BIC recommendation policy $\pi$ such that $\Pr[\pi(s) = m] > 0$. The set of all such menus is denoted as $\EX^{\delta}_s[\S]$. The $\delta$-signal-explorable set is defined as $\EX^{\delta}[\S] = \EX^{\delta}_S[\S]$. We omit $\delta$ in $\EX^{\delta}[\S]$ when $\delta = 0$.
\end{definition}

\xhdr{Approximate Information Monotonicity.}
In the following definition, we define a way to compare two signals approximately.
\begin{definition}
Let $S$ and $S'$ be two random variables. We say random variable $S$ is $\alpha$-approximately informative as random variable $S'$ about state $\omega_0$ if $I(S' ; \omega_0|S) = \alpha$.
\end{definition}

\begin{lemma}
\label{lem:ainfomono}
Let $S$ and $S'$ be two random variables and $\S$ and $\S'$ be their signal structures. If $S$ is $(\delta^2/8)$-approximately informative as $S'$ about state $\omega_0$ (i.e. $I(S' ; \omega_0|S) \leq \delta^2/8$), then $\EX_{s'}[\S'] \subseteq \EX^{\delta}_s[\S]$  for all $s' ,s$ such that $\Pr[S= s, S'= s'] > 0$.
\end{lemma}

\begin{proof}
For each signal realization $s$, denote
\[ D_s =  \DKL\left(\; ((S',\omega_0) \mid S=s)\quad \|\quad  (S'|S=s) \times (\omega_0\mid S=s) \;\right).  \]
We have
$\sum_{s} \Pr[S=s] \cdot D = I(S' ; \omega_0|S) \leq \delta^2/8$.

By Pinsker's inequality, we have
\begin{align*}
       &\sum_{s} \Pr[S = s]\cdot \sum_{s', \omega} | \Pr[S' = s', \omega_0 = \omega| S= s]\\
       & \qquad \qquad\qquad\qquad- \Pr[S'=s'|S=s] \cdot \Pr[\omega_0 = \omega|S=s]| \\
&\qquad\leq\textstyle   \sum_{s} \Pr[S=s] \cdot \sqrt{2 \ln(2)\cdot D_s  } \\
&\qquad\leq\textstyle   \sqrt{2 \sum_{s} \Pr[S=s] \cdot  D_s }
\leq  \delta /2.
\end{align*}

Consider any BIC recommendation policy $\pi'$ for signal structure $\S'$. We construct $\pi$ for signature structure $\S$ by setting
\[ \textstyle \Pr[\pi(s) = m] = \sum_{s'} \Pr[\pi'(s') = m] \cdot Pr[S' = s'|S = s].\]

Now we check $\pi$ is $\delta$-BIC. For any $m,m' \in \M$ and $\theta \in \varTheta$,
\begin{align*}
& \textstyle \sum_{\omega,s}\;
    \Pr[\omega_0= \omega] \cdot \Pr[S = s | \omega_0 = \omega] \\
&\quad\qquad \cdot \left(u(\theta, m(\theta), \omega) - u(\theta, m'(\theta), \omega)\right)
\cdot  \Pr[\pi(s) = m] \\
&\quad=
    \textstyle \sum_{\omega,s,s'}\;
        \Pr[\omega_0 = \omega, S = s] \cdot \Pr[ S'=s'|S= s] \cdot \Pr[\pi'(s') = m]   \\
&\qquad\qquad\cdot \left(u(\theta, m(\theta), \omega) - u(\theta, m'(\theta), \omega)\right)\\
&\quad\geq \textstyle \sum_{\omega,s,s'}\;
 \Pr[\omega_0 = \omega, S = s, S'=s'] \cdot \Pr[\pi'(s') = m]   \\
&\qquad\qquad\qquad\cdot \left(u(\theta, m(\theta), \omega) - u(\theta, m'(\theta),
 \omega)\right)\\
&\qquad\qquad -2 \cdot \sum_{\omega,s,s'} | \Pr[\omega_0 = \omega, S = s] \cdot \Pr[ S'=s'|S= s]\\
&\qquad\qquad -  \Pr[\omega_0 = \omega, S = s, S'=s']| \\
&\quad= \sum_{\omega,s'} \Pr[\omega_0 = \omega, S'=s'] \cdot \Pr[\pi'(s') = m]  \\
&\qquad\qquad \cdot \left(u(\theta, m(\theta), \omega) - u(\theta, m'(\theta),
 \omega)\right)\\
&\qquad\qquad -2 \cdot \sum_{s} \Pr[S = s] \cdot  \sum_{s', \omega} | \Pr[S' = s', \omega_0 = \omega| S= s] \\
&\qquad\qquad- \Pr[S'=s'|S=s] \cdot \Pr[\omega_0 = \omega|S=s]| \\
&\quad\geq  0- 2 \cdot t\tfrac{\delta}{2} = ~-\delta.
\end{align*}

We also have for any $s', s ,m$ such that $\Pr[S' = s',S = s] >0 $ and $\Pr[\pi'(s') = m] >0$, we have $\Pr[\pi(s) = m] > 0$. This implies $\EX_{s'}[\S'] \subseteq \EX^{\delta}_s[\S]$.
\end{proof}

\xhdr{Max-Support Policy.}
We can solve the following LP to check whether a particular menu $m_0 \in\A$ is signal-explorable given a particular realized signal $s_0\in\X$. In this LP, we represent a policy $\pi$ as a set of numbers
    $x_{m,s} = \Pr[\pi(s)=m]$,
for each menu $m\in \M$ and each feasible signal $s\in \X$.

\begin{figure}[H]
\begin{mdframed}
\begin{alignat*}{2}
 & \textbf{maximize }    x_{m_0,s_0}\  \\
&  \textbf{subject to: }\\
 & \sum_{\omega \in \varOmega, s \in \X} \Pr[\omega] \cdot \Pr[s | \omega] &  & \cdot \left(u(\theta, m(\theta), \omega) - u(\theta, m'(\theta), \omega) + \delta\right)\\
    &\qquad\qquad \cdot x_{m,s'} \geq 0  &\ & \forall m,m' \in \M, \theta \in \varTheta \\
& \textstyle  \sum_{m\in \M}\; x_{m,s} = 1,  \ &\ & \forall s \in \X \\
& x_{m,s} \geq 0,  \ &\ & \forall s \in \X, m\in \M
\end{alignat*}
\end{mdframed}
\label{fig:nocc_lp}
\end{figure}

Since the constraints in this LP characterize any $\delta$-BIC recommendation policy, it follows that menu $m_0$ is $\delta$-signal-explorable given realized signal $s_0$ if and only if the LP has a positive solution. If such solution exists, define recommendation policy $\pi = \pi^{m_0,s_0}$ by setting $\Pr[\pi(s) = m] = x_{m,s}$ for all $m \in \M, s \in \X$. Then this is a $\delta$-BIC recommendation policy such that $\Pr[\pi(s_0) = m_0] > 0$.

\begin{definition}
Given a signal structure $\S$, a recommendation policy $\pi$ is called the $\delta$-max-support policy if $\forall s \in \X$  and $\delta$-signal-explorable menu $m\in \M$ given $s$, $\Pr[\pi(s) = m] > 0$.
\end{definition}

It is easy to see that we obtain $\delta$-max-support recommendation policy by averaging the $\pi^{m,s}$ policies define above.
Specifically, the following policy is a $\delta$-BIC and $\delta$-max-support policy.
\begin{align}
\label{eq:pimax2}
\pi^{max} = \frac{1}{|\X|} \sum_{s \in \X} \frac{1}{|\EX_s^{\delta}[\S]|} \sum_{m \in \EX_s^{\delta}[\S]} \pi^{m,s}.
\end{align}

\xhdr{Maximal Exploration.}
Let us design a subroutine, called  MaxExplore, which outputs a sequence of $L$ menus. We are going to assume $L \geq \max_{m,s} \frac{B_m(\gamma_0)}{ \Pr[\pi^{max}(s)=m]}$. $\gamma_0$ is defined in Algorithm \ref{alg:nocc_main} of Section \ref{sec:private_main}. $B_m$ is defined in Lemma \ref{lem:deltam}.

The goal of this subroutine MaxExplore is to make sure that for any signal-explorable menu $m$, $m$ shows up at least $B_m(\gamma_0)$ times in the sequence with probability exactly 1. On the other hand, we want that the menu of each specific location in the sequence has marginal distribution same as $\pi^{max}$.

 \begin{algorithm}[H]
    \caption{Subroutine MaxExplore}
    	\label{alg:nocc_explore}
    \begin{algorithmic}[1]
	\STATE \textbf{Input:} signal $S$, signal structure $\S$.
	\STATE \textbf{Output:} a list of menus $\mu$
	\STATE Compute $\pi^{max}$ as per \eqref{eq:pimax2}.
		\STATE Initialize $Res = L$.
		\FOR {each menu $ m \in \M$}
			\STATE $C_m \leftarrow L \cdot \Pr[\pi^{max}(S) = m]$.
                     		\STATE Add $\lfloor C_m\rfloor$ copies of menu $m$ into list $\mu$.
			\STATE $Res \leftarrow Res -\lfloor C_m \rfloor $.
			\STATE $p^{Res}(m)\leftarrow  C_m -  \lfloor C_m\rfloor$
		\ENDFOR
		\STATE $p^{Res}(m) \leftarrow p^{Res}(m) / Res$, $\forall m \in \M$.
		\STATE Sample $Res$ many menus from distribution $p^{Res}$ independently and add these menus into $\mu$.
		\STATE Randomly permute the menus in $\mu$.
	\RETURN $\mu$.	
     \end{algorithmic}
\end{algorithm}

Similarly as the MaxExplore in Section \ref{sec:public}, we have the following:
\begin{claim}
\label{clm:maxexplore_nocc}
Given realized signal $S$, MaxExplore outputs a sequence of $L$ menus. Each menu in the sequence marginally distributed as $\pi^{max}$. For any menu $m$ such that $\Pr[\pi^{max} = m] >0$, $m$ shows up in the sequence at least $B_m(\gamma_0)$ times with probability exactly 1. MaxExplore runs in time polynomial in $L$, $|\M|$, $|\varOmega|$, $|\X|$ (size of the support of the signal).
\end{claim}

\xhdr{Menu Exploration.}
If an agent in a given round follows a given menu $m$,  an action-reward pair is revealed to the algorithm after the round. Such action-reward pair is called a \emph{sample} of the menu $m$. Let $D_m(\omega)$ denote the distribution of this action-reward pair for a fixed state $\omega$
(with randomness coming from the agent arrivals).

We compute an estimate $\Delta_m$ of $D_m(\omega_0)$. This estimate is a \emph{triple-list}: an explicit list of (action, reward, positive probability) triples.

\begin{lemma}
\label{lem:deltam}
For any $\gamma > 0$, we can compute a triple-list $\Delta_m$ which is a function of
    $B_m(\gamma) = O\left(\ln 1/\gamma \right)$
samples of menu $m$ such that
\[
\forall \omega\in\varOmega \quad
\Pr[\Delta_m \neq D_m(\omega) \mid \omega_0 = \omega] \leq \gamma.
\]
\end{lemma}

\begin{proof}
Let $U$ be the union of the support of $D_m(\omega)$ for all $\omega \in \varOmega$. For each $u \in U$ ($u$ is just a sample of the menu), define
    \[ q(u,\omega) = \Pr_{v \sim D_m(\omega)}[v = u].\]
Let $\delta_m$ be small enough such that for all $\omega, \omega'$ with $D_m(\omega) \neq D_m(\omega')$, there exists $u \in U$, such that $|q(u,\omega) - q(u,\omega')| > \delta_m$.

Now we compute $\Delta_m$ as follows: Take $B_m(\gamma) = \frac{2}{\delta_m^2}\ln\left(\frac{2|U|}{\gamma}\right) $ samples and set $\hat{q}(u)$ as the empirical frequency of seeing $u$. And set $\Delta_m$ to be some $D_m(\omega)$ such that for all $u \in U$, $|q(u,\omega) - \hat{q}(u)| \leq \delta_m / 2$. Notice that if such $\omega$ exists, $\Delta_m$ will be unique. If no $\omega$ satisfies this, just pick $\Delta_m$ to be an arbitrary $D_m(\omega)$.

Now let's analyze $\Pr[\Delta_m \neq D_m(\omega)]$. Let's fix the state $\omega_0 = \omega$. By Chernoff bound, for each $u \in U$,
\[
\Pr[|q(u,\omega) -\hat{q}(u)| > \delta_m/2] \leq 2\exp\left(-2 \cdot
    (\delta_m/2)^2
    \cdot B_m(\gamma)\right) \leq \gamma/|U|.
\]
By union bound, with probability at least $1-\gamma$, we have for all $u \in U$, $|q(u,\omega) - \hat{q}(u)| \leq \delta_m / 2$. This implies $\Delta_m = D_m(\omega)$.
\end{proof}

\subsection{Main Recommendation Policy}
\label{sec:private_main}
In this subsection, we develop our main recommendation policy, Algorithm \ref{alg:nocc_main} (see pseudo-code), which explores all the eventually-explorable menus and then recommends the agents the best menu given all history. We pick $L$ to be at least
\[ \max_{m,s:\Pr[\pi(s)=m] >0} \frac{B_m(\gamma_0)}{ \Pr[\pi(s)=m]}\]
for all $\pi$ that might be chosen as $\pi^{max}$ by Algorithm \ref{alg:nocc_main}.

 \begin{algorithm}[t]
    \caption{Main procedure for private types }
    	\label{alg:nocc_main}
    \begin{algorithmic}[1]
    	\STATE {\bf Initialize:} signal $S_1 = \S_1= \perp$, phase $l=1$.
        \STATE \COMMENT{ $S_l$ and $\S_l$ are the signal and signal structure in phase $l$. }
    	\STATE Set $\gamma_1 =\min\left(\frac{\delta^2}{16|\M|\log(|\varOmega|)},\left( \frac{\delta^2}{32|M|}\right)^2\right)$ and $\gamma_2 =  \frac{1}{T|\M|}$.
    \STATE Set $\gamma_0=\min(\gamma_1,\gamma_2)$.
	\FOR {rounds $t=1$ to $T$}
		\IF {phase $l\leq |\M|$}
		\STATE \COMMENT{\textbf{Exploration}}
		\IF {$t \equiv 1 \pmod L$}
			\STATE Start a new phase:
            \STATE $\mu \leftarrow $ MaxExplore($S_l, \S_l$)
                \TAB\COMMENT{compute a list of $L$ menus}
		\ENDIF
		\STATE Suggest menu $\mu [ (t-1) \mod L + 1]$ to the agent.
		\IF {$t \equiv 0 \pmod L$}
			\STATE End of a phase:
			\FOR{each explored menu $m$ in the previous phase}
            \STATE use $B_m(\gamma_1)$ samples to compute $\Delta_m$ from Lemma \ref{lem:deltam}
            \ENDFOR
			\STATE {\bf If} no state $\omega\in \varOmega$ is consistent with $\Delta_m$ 
            (\ie $\Delta_m = D_m(\omega)$) for all explored menus $m$ {\bf then} \\
            \STATE \TAB pick any state $\omega$, and 
                set $\Delta_m \leftarrow D_m(\omega)$ for all explored menus $m$. \\
            \COMMENT{to ensure that \#signals is bounded by $|\varOmega|$.}
			\STATE $l \leftarrow l + 1$.
			\STATE Set $S_l = \{ \text{ $\Delta_m$:\; all explored menus $m$ } \}$.
            \STATE Set $\S_l$ to be the signal structure of $S_l$.
		\ENDIF
	\ELSE
		\STATE \COMMENT{\textbf{Exploitation}}
		\IF {this is the first exploitation round}
        \FOR {each menu $m$ explored during exploration}
            \STATE use $B_m(\gamma_2)$ samples to compute $\Delta_m$ from Lemma \ref{lem:deltam}.
        \ENDFOR
        \STATE Set $S_l = \{ \text{ $\Delta_m$:\; all explored menus $m$ } \}$.
        \ENDIF
		\STATE Suggest the menu which consists of the best action of each type conditioned on $S_l$ and the prior.
	\ENDIF
	\ENDFOR
     \end{algorithmic}
\end{algorithm}

It is easy to check by Claim \ref{clm:maxexplore_nocc} that for each agent, it is $\delta$-BIC to follow the recommended action if previous agents all follow the recommended actions. Therefore we have the following claim.
\begin{claim}
\label{clm:nocc_BIC}
Algorithm \ref{alg:nocc_main} is $\delta$-BIC.
\end{claim}

\begin{lemma}
\label{lem:exp_nocc}
For any $l > 0$, assume Algorithm \ref{alg:nocc_main} has at least $\min(l, |\M|)$ phases.
For a given state $\omega$, if a menu $m$ can be explored by a BIC recommendation policy $\pi$ at round $l$ (i.e. $ \Pr[\pi^l= m]> 0$), then such menu is guaranteed to be explored $B_m$ times by Algorithm \ref{alg:nocc_main} by the end of phase $\min(l, |\M|)$.
\end{lemma}

\begin{proof}
The proof is similar to Lemma \ref{lem:exp_public}. We prove by induction on $l$ for $l \leq |\M|$.


Let $S$ be the signal of Algorithm \ref{alg:nocc_main} in phase $l$. Let $S'$ be the history of $\pi$ in the first $l-1$ rounds. More precisely, $S' = R, H_1,...,H_{l-1}$. Here $R$ is the internal randomness of $\pi$ and
\[ H_t = (M_t, A_t,u(\Theta_t, M_t(\Theta_t), \omega_0))\]
is the menu and the action-reward pair in round $t$ of $\pi$.

Let $\M'$ to be the set of menus explored in the first $l-1$ phases of Algorithm \ref{alg:nocc_main}. By the induction hypothesis, we have $\forall t\in[l-1]$, $M_t \subseteq \M'$. Then:
\begin{align*}
I(S'; \omega_0| S) 
    &= I(R,H_1,...,H_{l-1}; \omega_0| S)  \\
    &= I(R; \omega_0| S) + I(H_1,...,H_{l-1}; \omega_0|S, R) \\
    &= I(H_1,...,H_{l-1}; \omega_0|S, R).
\end{align*}

By the chain rule of mutual information, we have
\begin{align*}
I(H_1,...,H_{l-1}; \omega_0|S, R) 
 = I(H_1;\omega_0|S,R) + I(H_2;\omega_0|S, R ,H_1) + \cdots + I(H_{l-1}; \omega_0|S,R,H_1,...,H_{l-2}).
\end{align*}

For all $t \in [l-1]$, we have
\begin{align*}
I(H_t; \omega_0 \mid S,R,H_1,...,H_{t-1}) 
&= I(M_t,A_t, u(\Theta_t, M_t(\Theta_t), \omega_0); \omega_0 \mid S,R,H_1,...,H_{t-1}) \\
&= I(A_t, u(\Theta_t, M_t(\Theta_t), \omega_0); \omega_0  \mid  S,R,H_1,...,H_{t-1}, M_t)\\
&\leq I(D_{M_t}; \omega_0 \mid S,R,H_1,...,H_{t-1},M_t).
\end{align*}
The second last step comes from the fact that $M_t$ is a deterministic function of $R,H_1,...,H_{t-1}$. The last step comes from the fact that $(A_t,u(\Theta_t, M_t(\Theta_t), \omega_0))$ is independent with $\omega_0$ given $D_{M_t}$.

Then we have
\begin{align*}
I(D_{M_t}; \omega_0 \mid S,R,H_1,...,H_{t-1},M_t)
&= \textstyle   \sum_{m \in \M'}\; \Pr[M_t = m] \cdot I(D_m;\omega_0  \mid  S,R,H_1,...,H_{t-1},M_t = m)\\
&\leq \textstyle  \sum_{m \in M'}\; \Pr[M_t = m] \cdot I(D_m;\omega_0 \mid  \Delta_m, M_t =m).\\
&\leq \textstyle  \sum_{m \in M'}\; \Pr[M_t = m] \cdot H(D_m \mid  \Delta_m, M_t =m).
\end{align*}
The last step comes from the fact that
\[ I(D_m; (S\backslash \Delta_m),R,H_1,...,H_{t-1} \mid \omega_0, \Delta_m, M_t =m) = 0.\]
By Lemma \ref{lem:deltam}, we know that $\Pr[D_m \neq \Delta_m \mid M_t = m] \leq \gamma_1$. By Fano's inequality, we have
\begin{align*}
&H(D_m \mid  \Delta_m, M_t =m) \leq H(\gamma_1) + \gamma_1 \log(|\varOmega| - 1) \\
&\leq 2\sqrt{\gamma_1} + \gamma_1 \log(|\varOmega| - 1) \leq \tfrac{\delta^2}{16|\M|}+\tfrac{\delta^2}{16|\M|}  = \tfrac{\delta^2}{8|\M|}.
\end{align*}

Therefore we have
\[
I(H_t; \omega_0 \mid S,R,H_1,...,H_{t-1}) \leq \tfrac{\delta^2}{8|\M|}, \forall t \in [l-1].
\]
Then we get
    $ I(S'; \omega_0 \mid S) \leq \delta^2/8$.

By Lemma \ref{lem:ainfomono}, we know that $\EX_{s'}[\S'] \subseteq \EX^{\delta}_s[\S]$. By Claim \ref{clm:maxexplore_nocc}, we know that phase $l$ will explore menu $m$ at least $B_m(\gamma_0)$ times.

When $l > |\M|$, we use the same argument as the last paragraph of the proof of Lemma \ref{lem:exp_public}.
\end{proof}

\OMIT{
\begin{corollary}[Restatement of Theorem \ref{thm:private_nocc}]
\label{cor:private_nocc}
For any $\delta > 0$, we have a $\delta$-BIC recommendation policy of $T$ rounds with expected total reward at least $\left(T - C\cdot \log(T) \right) \cdot \OPT$ for some constant $C$ which does not depend on $T$.
\end{corollary}} 

\begin{proof}[Proof of Theorem \ref{thm:private_nocc}]
By Claim \ref{clm:nocc_BIC}, Algorithm \ref{alg:nocc_main} is $\delta$-BIC.

By Lemma \ref{lem:exp_nocc}, for each state $\omega$, Algorithm \ref{alg:nocc_main} explores all the eventually-explorable menus (i.e. $\MExp_{\omega}$) by the end of $|\M|$ phases.

After that, by Lemma \ref{lem:deltam} and $\gamma_2 = \frac{1}{T|\M|}$, for a fixed state $\omega$, we know that with probability $1- 1/T$, $\delta_m = D_m$ for all $m \in \MExp_{\omega}$. In this case, the agent of type $\theta$ gets expected reward at least $u(\theta,m^*(\theta),\omega)$ where menu $m^* =\arg\max_{m \in \MExp_{\omega}} \sum_{\theta \in \varTheta} \Pr[\theta] \cdot u(\theta, m(\theta), \omega)$. Taking average over types, the expected reward per round should be at least $(1-1/T) \cdot \max_{m \in \MExp_{\omega}} \sum_{\theta \in \varTheta} \Pr[\theta] \cdot u(\theta, m(\theta), \omega)$.

The expected number of rounds of the first $|\M|$ phases is $|\M| \cdot L = O(\ln(T))$. Therefore, Algorithm \ref{alg:nocc_main} has expected total reward at least $T \cdot \OPT- T \cdot (1/T) - O(\ln(T)) = T\cdot \OPT - O(\ln(T))$.
\end{proof}

\appendix
\section{Basics of Information Theory}
\label{app:info-theory}

We briefly review some standard facts and definitions from information theory which are used in proofs. For a more detailed introduction, see \cite{CK11}. Throughout, $X,Y,Z,W$ are random variables that take values in an arbitrary domain (not necessarily $\R$).

\xhdr{Entropy.}
The fundamental notion is \emph{entropy} of a random variable. In particular, if $X$ has finite support, its entropy is defined as
\[ H(X) = \textstyle - \sum_{x} p(x)\cdot  \log p(x),
\quad\text{where } p(x) = \Pr[X = x]. \]
(Throughout this paper, we use $\log$ to refer to the base $2$ logarithm and use $\ln$ to refer to the natural logarithm.) If $X$ is drawn from Bernoulli distribution with $\E[X]=p$, then
    \[ H(p) = -(p\log p + (1-p)(\log(1-p)). \]

The conditional entropy of $X$ given event $E$ is the entropy of the conditional distribution $(X|E)$:
\[ H(X|E) = \textstyle - \sum_{x} p(x)\cdot  \log p(x),
\quad\text{where } p(x) = \Pr[X = x | E]. \]

The \emph{conditional entropy} of $X$ given $Y$ is
\[ H(X|Y)
    := \textstyle \E_y[H(X|Y = y)]
    = \textstyle \sum_{y} \Pr[Y=y]\cdot H(X|Y = y). \]
Note that $H(X|Y) = H(X)$ if $X$ and $Y$ are independent.

We are sometimes interested in the entropy of a tuple of random variables, such as $(X,Y,Z)$. To simplify notation, we write $H(X,Y,Z)$ instead of $H((X,Y,Z))$, and similarly in other information-theoretic notation. Now, we formulate the \emph{Chain Rule} for entropy:
\begin{align}\label{app:info-entropy-chain-rule}
 H(X,Y) = H(X) + H(Y|X).
\end{align}

We also use the following fundamental fact about entropy:

\begin{lemma}[Fano's Inequality]
Let $X,Y,\hat{X}$ be random variables such that $\hat{X}$ is a deterministic function of $Y$.%
\footnote{Informally, $\hat{X}$ is an approximate version of $X$ derived from signal $Y$.} Let $E = \{ \hat{X} \neq X \}$ be the ``error event". Then, letting
$\X$ denote the support set of $X$,
    \[ H(X|Y) \leq H(E) + \Pr[E] \cdot (\log(|\X|-1), \]
\end{lemma}

\xhdr{Mutual info.}
The \emph{mutual information} between $X$ and $Y$ is
\[ I(X;Y) := H(X) - H(X|Y) = H(Y) - H(Y|X).\]
The \emph{conditional mutual information} between $X$ and $Y$ given $Z$ is
\[ I(X;Y|Z) := H(X|Z) - H(X|Y,Z) = H(Y|Z) - H(Y|X,Z).\]
Note that $I(X;Y|Z) = I(X;Y)$ if $X,Z$ are conditionally independent given $Y$, and $Y,Z$ are conditionally independent given $X$.

Important properties of conditional mutual information are:
\begin{align}
I(X,Y;Z|W) &= I(X;Z|W) + I(Y;Z|W,X) \\
I(X;Y|Z) &\geq I(X;Y|Z,W) \qquad\text{if $I(Y;W|X,Z) = 0$} \\
I(X;Y|Z) &\leq I(X;Y|Z,W) \qquad\text{if $I(Y;W|Z) = 0$}
\end{align}

\xhdr{KL-divergence.}
The \emph{Kullback-Leibler divergence} (a.k.a., \emph{KL-divergence}) between random variables $X$ and $Y$ is defined as
\[ \DKL(X\| Y) = \sum_x \Pr[X = x]
    \cdot \log\left( \frac{\Pr[X = x]}{\Pr[Y = x]} \right) .\]
Note that the definition is not symmetric, in the sense that in general
    $\DKL(X\| Y)\neq \DKL(Y\| X)$.
KL-divergence can be related to conditional mutual information as follows:
\begin{align}
&I(X;Y|Z)
    = \mathbb{E}_{x,z}\left[ \; \DKL((Y|X = x, Z=z)\|(Y|Z=z)) \; \right] \nonumber \\
&\;= \sum_{x,z} \Pr[X=x,Z=z]\;\ \DKL((Y|X = x, Z=z)\|(Y|Z=z)).
\end{align}
Here $(Y|E)$ denotes the conditional distribution of $Y$ given event $E$.

We also use \emph{Pinsker Inequality}:
\begin{align}
\textstyle \sum_x | \Pr[X=x] - \Pr[Y=x]| \leq \sqrt{2 \ln(2)\, \DKL(X\|Y)}.
\end{align}

\newpage
\bibliographystyle{ACM-Reference-Format}
\bibliography{references,bib-abbrv,bib-slivkins,bib-bandits,bib-AGT,bib-ML}

\end{document}